\documentclass[11pt,a4paper,english]{article}
\usepackage{amsmath,amsthm,amsfonts,amscd,eucal,latexsym,amssymb,mathrsfs,cancel,slashed,color}
\usepackage{hyperref}
\usepackage{url}
   \newtheorem{thm}{Theorem}[section]
          \newtheorem{prop}[thm]{Proposition}
          \newtheorem{lem}[thm]{Lemma}
          \newtheorem{cor}[thm]{Corollary}

          \newtheorem{defn}[thm]{Definition}

          \newtheorem{rem}[thm]{Remark}

\input{xy}
\xyoption{all}

\input{epsf}

\usepackage[paperwidth=18.6cm, paperheight=25cm, margin=2.2cm]{geometry}

\catcode`@=11
\renewcommand{\section}
{\@startsection{section}{1}{0pt}{\medskipamount}{\medskipamount}{\large\bf}}
\makeatletter\renewcommand{\subsection}
{\@startsection{subsection}{2}{\z@}{-3.25ex plus -1ex minus -.2ex}
{1.5ex plus .2ex}{\bf }}

\numberwithin{equation}{section}
\allowdisplaybreaks[1]

\catcode`@=12

\def\cD{\mathcal{D}}

\def\cH{\mathcal{H}}
\def\cN{\mathcal{N}}
\def\cM{\mathcal{M}}

\def\cS{\mathcal{S}}
\def\cT{\mathcal{T}}

\def\mM{\mathscr{M}}
\def\mR{\mathscr{R}}

\def\fA{\mathfrak{A}}
\def\fI{\mathfrak{I}}
\def\fF{\mathfrak{F}}
 \usepackage{relsize}

\newcommand{\MassMultiplicativeOperator}{\cT}
\newcommand{\Dirac}[1]{\cD_{#1}}
\newcommand{\propagator}[2]{E^{#1}_{#2}}
\newcommand{\spinorpropagator}[2]{E^{#1}_{#2}}
\newcommand{\cospinorpropagator}[2]{E^{*\,#1}_{#2}}
\newcommand{\FermionicOperator}{\cS}
\newcommand{\FermionicProjectionOperator}[1]{\cS_{#1}}
\newcommand{\ModifiedFermionicProjectionOperator}[1]{{\mathlarger{\mathlarger{\mathlarger{ \mathfrak{s}}}}}_{#1}}
\newcommand{\SmearingMassOperator}{\mathfrak{p}}
\newcommand{\Moller}[1]{R_{#1}}
\newcommand{\SequenceMoller}{\mathfrak{R}_{\massofinterest}}
\newcommand{\projector}{\Pi}
\newcommand{\FermionicProjectorState}{\omega_{\textrm{FP}}}

\newcommand{\fiberpairing}[2]{\prec#1\ |\ #2\succ}
\newcommand{\scalarproduct}[2]{\left(#1\ |\ #2\right)}
\newcommand{\spinorscalarproduct}[2]{\left(#1\ |\ #2\right)_{(s)}}
\newcommand{\cospinorscalarproduct}[2]{\left(#1\ |\ #2\right)_{(c)}}
\newcommand{\directsumscalarproduct}[2]{\left(#1\ |\ #2\right)_{\massinterval}}

\newcommand{\bracket}[2]{\left\langle#1\ |\ #2\right\rangle}
\newcommand{\innerproduct}[2]{\left\langle#1\ |\ #2\right\rangle_{(s)}}
\newcommand{\MassSmearedinnerproductOperatorNotation}{\cN}

\newcommand{\MassSmearedinnerproduct}[2]{\left\langle\SmearingMassOperator#1\ |\ \SmearingMassOperator#2\right\rangle_{(s)}}
\newcommand{\norm}[2]{\|#1\|_{#2}}
\newcommand{\spacetime}{\cM}
\newcommand{\Hilbert}[1]{\cH_{#1}}
\newcommand{\Hilbertspinor}[1]{\cH_{(s), #1}}
\newcommand{\Hilbertcospinor}[1]{\cH_{(c), #1}}
\newcommand{\HilbertDirectSum}{\cH_{(s),\massinterval}}

\newcommand{\sol}[1]{\textsf{Sol}(\Dirac{#1})}
\newcommand{\dualsol}[1]{\textsf{Sol}(\adjointWRTfiberpairing{\Dirac{#1}})}
\newcommand{\soldirectsum}{\textsf{Sol}_{\massinterval}}
\newcommand{\labelspace}[1]{\mathcal{S}ol_#1}

\newcommand{\sections}[2]{\Gamma_{\textrm{#1}}(#2)}
\newcommand{\spinorbundle}{S\spacetime}
\newcommand{\cospinorbundle}{S^*\spacetime}

\newcommand{\cotangentbundle}{T^*\spacetime}

\newcommand{\algebra}[1]{\fA_{#1}(\spacetime,g)}
\newcommand{\algebraF}[1]{\fF_{#1}(\spacetime,g)}
\newcommand{\ideal}{\fI}

\newcommand{\volumeform}{\textrm{d}\mu_g}
\newcommand{\Cauchysurface}{\Sigma}
\newcommand{\volumeformSigma}{\textrm{d}\mu_\Sigma}
\newcommand{\gammatrix}[1]{\gamma_{#1}}

\newcommand{\anticommutator}[2]{\lbrace #1,#2\rbrace}
\newcommand{\adjunction}[1]{A#1}
\newcommand{\spinorchargeconjugation}[1]{C_{(s)}}
\newcommand{\cospinorchargeconjugation}[1]{C_{(c)}}

\newcommand{\trace}[2]{\textrm{Tr}_{#1}({#2})}
\newcommand{\spinoroperator}{\nabla_{(s)}}
\newcommand{\cospinoroperator}{\nabla_{(c)}}
\newcommand{\genericmass}{\beta}
\newcommand{\interpolatingmass}{m}
\newcommand{\massofinterest}{\alpha}
\newcommand{\adjointWRTfiberpairing}[1]{#1^{\;\star}}
\newcommand{\Compact}{\textrm{c}}
\newcommand{\TimelikeCompact}{\textrm{tc}}
\newcommand{\PastCompact}{\textrm{pc}}
\newcommand{\FutureCompact}{\textrm{fc}}
\newcommand{\SpacelikeCompact}{\textrm{sc}}
\newcommand{\involution}{\Upsilon}
\newcommand{\unit}[1]{1_{\algebraF{\genericmass}}}
\newcommand{\state}{\omega}

\newcommand{\massinterval}{I_\massofinterest}
\newcommand{\dvol}[1]{\textrm{d}#1}
\newcommand{\domain}[1]{D(#1)}
\newcommand{\spinorm}[1]{\norm{#1}{(s)}}
\newcommand{\SequenceSpinorm}[1]{\norm{#1}{(s),\massinterval}}
\newcommand{\MassFunction}[1]{\mathfrak{m}_{#1}}

\DeclareFontFamily{OT1}{rsfso}{}
\DeclareFontShape{OT1}{rsfso}{m}{n}{ <-7> rsfso5 <7-10> rsfso7 <10-> rsfso10}{}
\DeclareMathAlphabet{\mycal}{OT1}{rsfso}{m}{n}
\newcommand{\Rindler}{\mR} % {\mycal R} %{\mR}
\newcommand{\Minkowski}{\mM} %{\mycal M}%{\mM} 
\newcommand{\complex}{\mathbb{C}}
\def\real{\mathbb{R}}
\def\natural{\mathbb{N}}

\newcommand{\supp}[1]{\textrm{supp}(#1)}
\newcommand{\WaveFrontSet}[1]{\textrm{WF}(#1)}

\newcommand{\inverse}[1]{#1^{-1}}
\definecolor{NiColor}{RGB}{77,77,255}

\begin{document}

\begin{flushright}

\baselineskip=6pt

\end{flushright}

\begin{center}
\vspace{5mm}

{\Large\bf A new class of Fermionic Projectors:} \\ \vspace{3mm} {\Large \bf M\o ller operators and mass oscillation properties}

\vspace{8mm}

{\bf Nicol\'o Drago}${}^{\, 1,a}$ \ and \ {\bf Simone Murro}${}^{\, 2,b}$
\\[4mm]
\noindent  ${}^1$ {\it Dipartimento di Matematica, }{\it Universit\`a di Genova, }
{\it I-16146 Genova, Italy.}
\\[2mm]
\noindent  ${}^2$ {\it Fakult\"at f\"ur Mathematik, }{\it Universit\"at Regensburg, } {\it D-93040 Regensburg, Germany.}
\\[3mm]
Email: \ {\tt ${}^a\,$drago@dima.unige.it,  ${}^b\,$simone.murro@ur.de.}
\\[8mm]
October 1, 2016
\\[10mm]
\end{center}

\begin{abstract}
Recently, a new functional analytic construction of quasi-free states for a self-dual CAR algebra has been presented in \cite{Felix2}. This method relies on the so-called strong mass oscillation property. We provide an example where this requirement is not satisfied, due to the nonvanishing trace of the solutions of the Dirac equation on the horizon of Rindler space, and we propose a modification of the construction in order to weaken this condition. Finally, a connection between the two approaches is built.
\end{abstract}

\\[0.5mm]
%\tableofcontents

\renewcommand{\thefootnote}{\arabic{footnote}}
\setcounter{footnote}{0}

\section{Introduction}\label{Introduction}

The algebraic approach is a mathematically rigorous scheme which is especially well-suited for formulating quantum theories also on globally hyperbolic spacetimes -- see \cite{Fulling, Haag, gerard, aqft2} for textbooks, \cite{aqft1, ThomasAlex, FR16} for recent reviews and \cite{DMP2, DMP3, SDH12, cospi, nico2, dual} for some applications.
The quantization of a theory in the algebraic approach is based on two steps.
The first consists of the assignment to a physical system of a $*$-algebra of observables which encodes structural properties such as causality, dynamics and the canonical commutation/anticommutation relations.
The second step calls for the identification of an algebraic state, which is a positive, linear and normalized functional on the algebra of observables.
The image of an algebra element under the action of a state is interpreted as the mean value of the associated observable on that state: In this way, it is possible to recover the usual probabilistic interpretation proper of quantum theories.
However, not every state can be regarded as physically relevant.
It is widely accepted that a criterion to single out the physical ones is to require the so-called Hadamard condition \cite{GK,wald,ChrisVerch}
The reasons for this choice are manifold: For example, it implies the finiteness of the quantum fluctuations of the expectation value of every observable and it allows to construct Wick polynomials following a covariant scheme, see \cite{HW} or \cite{IgorValter} for recent reviews.
Thanks to the seminal work of Radzikowski \cite{Radzikowski,Radzikowski2}, the Hadamard condition has been translated into the language of microlocal analysis, as a constraint on the wavefront set of the bidistribution associated to the two-point function of the state.
Nowadays, several methods to construct Hadamard states are known, see for example \cite{Fulling1, Fulling2, DMP, GW1, GW2, BDM, WZ, nico}. Yet there are several, physically interesting scenarios which have not been analyzed. In a recent paper \cite{Felix2} a new functional analytic construction of quasi-free states for a massive Dirac field in a general class of globally hyperbolic spacetimes was proposed.
This method takes advantage of the results in \cite{araki}, in particular, the one proving that the construction of a projection operator in a Hilbert space is equivalent to the assignment of a pure, quasi-free state on a CAR $*$-algebra. In particular, the procedure calls for the identification of a specific operator, dubbed Fermionic Projector. In a few words, it works as follows: Consider both the Dirac equation with the mass $m$ varying parametrically and its smooth solutions, which are chosen to be spacelike compact as well as compact in $m$. Such set can be completed to Hilbert space with respect to the scalar product induced by integration both over the manifold and over the mass.

Subsequently we assume the so-called strong mass oscillation property, that is a constraint on the decay rate at infinity of the solutions of the massive Dirac equation after integration over the mass. As a by-product this leads to the identification on the space of such solutions a continuous sesquilinear form. By applying Riesz theorem, this is tantamount to the assignment of a family of bounded symmetric operators, each acting on the subspace of solutions with a specific, fixed value of the mass. The net advantage of this construction is its independence from any structural property of the underlying background, such as the existence of specific Killing fields. 
This method has been successfully applied in \cite{simo, Felix5} to construct a distinguished
Hadamard state for a Dirac field on Minkowski spacetime, in the presence of an external time-dependent potential, subject to suitable technical constraints. Despite these successes, an undeniable limitation of this method is the intrinsic difficulty of proving that the strong mass oscillation property holds true. It has to be checked case by case, and, in general, there exist spacetimes, e.g., Rinder, where it fails to be true.

The goal of this paper is thus twofold. On the one hand we investigate an alternative procedure aimed at weakening the strong mass oscillation property, on which the construction of the Fermionic Projector lies. On the other hand, we use this novel approach to identify a class of Hadamard states  for suitable spacetimes. 
Our main idea consists of individuating unitary operators that intertwine the dynamics of two Green hyperbolic operators differing only by a mass term, extending thus the work of \cite{nico1,nico}.
Such operators in combination with an  integration over the mass of the solutions  of the Dirac equation defines a new sesquilinear form. This is continuous in either one or both entries, whenever two modified versions of the mass oscillation property are satisfied.
Once more, on account of the Riesz representation theorem, such sesquilinear form yields on the Hilbert space, built out of the smooth solutions of the Dirac equation, a symmetric, linear operator that is unbounded when only the modified weak mass oscillation property holds, while it is bounded in the other one.
In addition, we construct a pure and quasi-free state on the CAR $*$-algebra associated to a Dirac field, still using the results of Araki \cite{araki}, and we discuss whether the Hadamard condition is satisfied.
To prove the robustness of our novel method,  we investigate in detail a concrete example in which the strong mass oscillation property does not hold true, but the modified weak one does: a massive Dirac field on Rindler spacetime.

The paper is structured as follows: In Section \ref{Algebraic approach to Dirac fields} we outline both the geometric setting necessary to present the Dirac equation and its space of solutions.
Subsequently, we quantize the Dirac field following the algebraic approach.
Section \ref{Fermionic Projector state} is focused on introducing the results of \cite{Felix2} while in Section \ref{Modified Fermionic Projector state}, we construct the new class of modified Fermionic Projector states. Finally, in Section \ref{Rindler counterexample} we outline an example in which only the modified weak mass oscillation property holds.
\section{Algebraic approach to quantum Dirac fields}\label{Algebraic approach to Dirac fields}

\subsection{Dirac operator}\label{Setting}
To make the present paper sufficiently self-contained, we summarize a
few basic structural properties of spinor fields in curved spacetimes. For more details, we refer to \cite{Isham, Lawson}.

Let $\spacetime$ be a {\bf spacetime}, namely a four dimensional, Hausdorff, connected, orientable and time orientable smooth manifold endowed with a smooth Lorentzian metric $g$ of signature $(+, - , -, -)$.
In order for the Cauchy problem of the Dirac equation to be well-posed, we consider only spacetimes which are \textbf{globally hyperbolic}, \textit{i.e.} they possess a \textbf{Cauchy surface} $\Sigma$. This is a codimension 1, achronal subset whose domain of dependence coincides with the whole manifold.  In view of the analysis in \cite{Isham}, the spinor bundle can be defined as the trivial vector bundle $\spinorbundle\doteq \spacetime\times\complex^4$. Its dual, the cospinor bundle, is $\cospinorbundle \doteq\spacetime\times (\complex^4)^*$. Given any vector bundle $E\stackrel{\pi}{\to}\spacetime$ of finite rank, we denote by $\sections{\Compact}{E},\sections{\SpacelikeCompact}{E}$, $\sections{\FutureCompact}{E}$, $\sections{\PastCompact}{E}$ and $\sections{\TimelikeCompact}{E}=\sections{\FutureCompact}{E}\cap\sections{\PastCompact}{E}$ the spaces of compact, space-/future-/past-/time-like compact elements respectively \cite{KO}.
Smooth global sections of $\spinorbundle$ and of $\cospinorbundle$ are called \textbf{spinors} and \textbf{cospinors} respectively. We will indicate with $\fiberpairing{}{}$ the pointwise fiber pairing between smooth sections of $\sections{}{\spinorbundle}$ and of $\sections{}{\cospinorbundle}$. 
In addition, we introduce the $\gamma$-matrices $\{\gammatrix{a}\}_{a=0,1,2,3},$ the choice of which corresponds to fixing an irreducible representation of the Clifford algebra $C\ell_{1,3}(\mathbb{C})$.
Since each $\gamma_a$ is invertible,  we define the \textbf{adjunction map}
\begin{gather}\label{adjunction map}
\adjunction{}:\spinorbundle\to\cospinorbundle\qquad (x,v)\mapsto(x, v^\dagger\gammatrix{0}),
\end{gather}
which turns out to be a complex anti-linear vector bundle isomorphism.
Here $\dagger$ indicates the operations of transpose and of conjugation. 
The covariant derivatives both for spinors and for cospinors 
can be introduced as the first order linear partial differential operators
\begin{gather*}\label{DO per Dirac}
\spinoroperator:\sections{}{\spinorbundle}\to\sections{}{\spinorbundle}\qquad
\spinoroperator\psi\doteq\trace{g}{\gamma\nabla\psi}\\
\cospinoroperator:\sections{}{\cospinorbundle}\to\sections{}{\cospinorbundle}\qquad
\cospinoroperator\phi\doteq\trace{g}{\nabla\phi\gamma} \, .
\end{gather*}
Here  $\text{Tr}_g$ denotes the metric-contraction of the covariant two-tensor $\gamma\nabla\psi$ taking values in $\sections{}{\spinorbundle}$. A similar definition applies to $\nabla\phi\gamma$ taking values in $\sections{}{\cospinorbundle}$.
We have now all necessary tools to introduce the \textbf{Dirac operator} $\Dirac{\genericmass}:\sections{}{\spinorbundle}\to\sections{}{\spinorbundle}$ 
\begin{gather*}\label{Dirac operator}
\Dirac{\genericmass}\psi\doteq i\spinoroperator\psi-\genericmass\psi \, ,
\end{gather*}
where $\genericmass\in C^\infty(\spacetime,\real)$.
We introduce the sesquilinear form
\begin{gather}\label{bracket}
\bracket{\psi}{\phi}\doteq\int_\spacetime\fiberpairing{\psi}{\phi}\volumeform,
\end{gather}
for any $\psi\in\sections{}{\spinorbundle},\ \phi\in\sections{}{\cospinorbundle}$ such that $\fiberpairing{\psi}{\phi}\in L^1(\spacetime)$. Equivalently, we call
  $\adjointWRTfiberpairing{\Dirac{\genericmass}}:\sections{}{\cospinorbundle} \to \sections{}{\cospinorbundle}$ the formal adjoint of $\Dirac{\genericmass}$, unambiguously defined via 
\begin{gather*}\label{ajoint of Dirac operator}
\bracket{\psi}{\adjointWRTfiberpairing{\Dirac{\genericmass}}\phi}=
\bracket{\Dirac{\genericmass}\psi}{\phi}\qquad\psi\in\sections{}{\spinorbundle},\phi\in\sections{}{\cospinorbundle},
\end{gather*}
which reads explicitly  $\adjointWRTfiberpairing{\Dirac{\genericmass}}\phi=-i\cospinoroperator\phi-\genericmass\phi$. 

\subsection{The space of solutions}\label{space of solutions}

In this section, we want to characterize the space of smooth, space-like compact solutions of the Dirac equation. For a complete analysis, we refer to \cite{nicolas, Bar2, waldmann, Bar1}.

Since the Dirac operator $\Dirac{\genericmass}$ is a \textbf{prenormally hyperbolic operator}  \cite{Wr}, \textit{i.e.} $\Dirac{\genericmass}\circ\Dirac{\genericmass}$ is of hyperbolic type, there exist unique advanced/retarded propagators, namely continuous operators $\spinorpropagator{\pm}{\genericmass}:\sections{pc/fc}{\spinorbundle}\to\sections{pc/fc}{\spinorbundle}$ such that  
\begin{gather*}
\spinorpropagator{\pm}{\genericmass}\circ\Dirac{\genericmass}=
\Dirac{\genericmass}\circ\spinorpropagator{\pm}{\genericmass}=\textrm{Id}_{\spinorbundle},\qquad
\supp{\spinorpropagator{\pm}{\genericmass}\psi}\subseteq J^\pm(\supp{\psi}),
\end{gather*}
where $J^\pm(\mathcal{O})$ denotes the causal future/past of $\mathcal{O}$.
Taking the difference between the advanced and the retarded propagator, we can define the \textbf{causal propagator} as $
\spinorpropagator{}{\genericmass}\doteq\spinorpropagator{+}{\genericmass}-\spinorpropagator{-}{\genericmass}\colon\sections{tc}{\spinorbundle}\to\sections{}{\spinorbundle}  \, .$
As a by-product, the space of smooth, space-like compact solutions of the Dirac equation, denoted by $\sol{\genericmass}\doteq\spinorpropagator{}{\genericmass}\big(\sections{\Compact}{\spinorbundle}\big)$ is isomorphic to
\begin{gather}\label{Characterization of the space of smooth space-like solutions}
\sol{\genericmass}\simeq
\frac{\sections{\Compact}{\spinorbundle}}{ \Dirac{\genericmass}\sections{\Compact}{\spinorbundle}}.
\end{gather}
Now, for any $\psi_1,\psi_2\in\sol{\genericmass}$, let $\spinorscalarproduct{\cdot}{\cdot}$ be the scalar product defined by
\begin{gather}\label{spinor scalar product}
\spinorscalarproduct{\psi_1}{\psi_2}\doteq
\int_\Cauchysurface\fiberpairing{\adjunction{\psi_1}}{\gamma^\mu{n_\mu}\psi_2}\volumeformSigma,
\end{gather}
where $\Cauchysurface$ is any, but fixed Cauchy surface with future pointing unit normal $n$.
Throughout the paper the subscript $(s)$ (resp. $(c)$) denotes spinor (resp. cospinor) quantities.
%For the rest of the paper the subscript $(s)$ denotes that all the variable in the game are spinors, while $(c)$ are cospinors.
Since $\psi_1,\ \psi_2\in \sol{\genericmass}$, the scalar product does not depend on the choice of $\Cauchysurface$ (we refer to \cite{aqft0} for the details).
We denote with $\Hilbertspinor{\genericmass}$ the Hilbert space obtained by completing $\sol{\genericmass}$ with respect to $\spinorscalarproduct{\cdot}{\cdot}$. By duality, the space $\dualsol{\genericmass}$ of space-like compact and smooth sections of $\cospinorbundle$ such that $\adjointWRTfiberpairing{\Dirac{\genericmass}}\phi=0$ is isomorphic to
\begin{equation}\label{iso2}
\dualsol{\genericmass}\simeq \sections{\Compact}{\cospinorbundle} / \adjointWRTfiberpairing{\Dirac{\genericmass}}\sections{\Compact}{\cospinorbundle}.
\end{equation}
Similarly we denote $\Hilbertcospinor{\genericmass}$ the Hilbert space obtained as the completion of $\dualsol{\genericmass}$ with respect to $\cospinorscalarproduct{\cdot}{\cdot} \doteq \spinorscalarproduct{A\cdot}{A\,\cdot}$. 

\subsection{Algebra of Dirac fields}
Having under control the dynamics of a classical Dirac fields, we are ready to quantize it. Our
quantization scheme is based on the so-called algebraic approach to quantum field theory, initially
developed by Haag and Kastler in Minkowski spacetime \cite{HK} and later extended to curved
backgrounds by Dimock \cite{Di80}. To introduce the algebraic approach to Dirac fields, we also
profit from\cite{dimock, FVD, aqft3, KOD, ZAD, aqft2}.

First of all we introduce the unital, universal tensor algebra over $\complex$, 
$$\algebra{\genericmass}\doteq\bigoplus_{n\in\natural}\labelspace{\genericmass}^{\otimes n}\, ,$$
where $\labelspace{\genericmass}^{\otimes 0}\equiv\complex$ and $
\labelspace{\genericmass}\doteq
\sol{\genericmass}\oplus\dualsol{\genericmass}$, equipped with the $*$-operation induced by the anti-unitary involution  
\begin{gather}\label{involution on the label space}
\involution: \labelspace{\genericmass}\to\labelspace{\genericmass} \qquad \involution(\psi\oplus\phi)\doteq A^{-1}{\phi}\oplus\adjunction{\psi} \, ,
\end{gather}
extended per anti-linearity to the whole $\algebra{\genericmass}$. Here $A$ is the adjunction map introduced in \eqref{adjunction map}.  $\algebra{\genericmass}$ can be equipped with a natural topology as follows:  We require that every $f_j=\oplus_n f_j^n$ converges to $\oplus_n f^n$ in $\labelspace{\genericmass}$ with respect to the limit Fr\'echet topology and that there exists $N\in \mathbb{N}$ such that $f_j^n$ vanishes for every $n > N$ and for every $j$. 
We equip $\labelspace{\genericmass}$ with the  scalar product 
$$ \scalarproduct{\psi_1\oplus\phi_1}{\psi_2\oplus\phi_2}\doteq \spinorscalarproduct{\psi_1}{\psi_2}+ \cospinorscalarproduct{\phi_1}{\phi_2} ,$$
for any $\psi_1, \psi_2 \in \sol{\genericmass} $ and $ \phi_1, \phi_2 \in \dualsol{\genericmass}$.  The \textbf{algebra of Dirac fields} on $\spacetime$ is the unital, topological $*$-algebra obtained by the quotient
\begin{gather}\label{off-shell algebra}
\algebraF{\genericmass}\doteq\frac{\algebra{\genericmass}}{\ideal},
\end{gather}
where $\ideal$ is the $*$-ideal generated by the abstract elements that satisfy the canonical anti-commutation relations 
$\anticommutator{\psi_1\oplus\phi_1}{\psi_2\oplus\phi_2}-\scalarproduct{\psi_1\oplus\phi_1}{\psi_2\oplus\phi_2}\unit{\genericmass}$, where $\unit{\genericmass}$ is the unit element of $\algebraF{\genericmass}$.
Notice that, since $\involution$ is anti-unitary, the action of its extension to $\algebra{\genericmass}$ descends to $\algebraF{\genericmass}$. The algebra of Dirac fields~\eqref{off-shell algebra} is an example of self-dual CAR algebra. For more details see \cite{araki}. 

\subsection{Quasi-free states}\label{quasi-free states}

The quest to find a representation of the algebra of Dirac fields on a suitable Hilbert
space is a difficult aspect of the algebraic approach to quantum field theory, which relies on the notion of a state. On $\algebraF{\genericmass}{}$ this is
nothing but a complex linear functional $\state:\algebraF{\genericmass}{}\to\complex$, which is normalized, $\omega(\unit{\genericmass}) = 1 $, and  positive, $\state(a^*a) \geq 0$ for all $a\in\algebraF{\genericmass}.$ Once a state has been fixed, the $*$-algebra can be represented in terms of linear operators
on a Hilbert space. This is, indeed, a consequence of the renown GNS representation theorem for unital $*$-algebras. We refer to \cite{DHP, IgorValter} for more details.
Due to the natural grading on $\algebraF{\genericmass}{}$, a state is specified once its $n$-points functions
\begin{gather*}\label{n-points functions}
\state_n( f_1, \ldots, f_n)\doteq\state(f_1\otimes\ldots\otimes f_n)\qquad f_j\in\labelspace{\genericmass}\qquad j=1,\ldots, n, 
\end{gather*}
are assigned.
We focus on \textbf{quasi-free (or Gaussian) states}, namely those  whose $n$-point functions vanish for odd $n$, while for even $n$, they are defined as
\[ \state_n( [f_1], \ldots, [f_n])=\sum\limits_{\sigma \in S'_n} (-1)^{\text{\rm{sign}}(\sigma)} \prod\limits_{i=1}^{n/2}
\;\omega_2 \big([f_{\sigma(2i-1)}], [f_{\sigma(2i)}] \big) \:, \]
where~$S'_n$ denotes the set of ordered permutations of $n$ elements.
On account of the isomorphisms \eqref{Characterization of the space of smooth space-like solutions} and \eqref{iso2}, we can uniquely associate to $\state_2$ a bidistribution in $\Big(\Gamma_c{\big(\spinorbundle\oplus\cospinorbundle\big)^2}\Big)'$ by means of  the relation
$$\widetilde{\state_2}(u\oplus v,u'\oplus v')\doteq
\state_2(\psi_u\oplus\phi_v,\psi_{u'}\oplus\phi_{v'}),
$$ 
where $\psi_{u^{(\prime)}}\doteq\spinorpropagator{}{\genericmass}u^{(\prime)}$, $\phi_{v^{(\prime)}}\doteq\cospinorpropagator{}{\genericmass}v^{(\prime)}$ where $u,u'\in\sections{\Compact}{\spinorbundle}$, $v,v'\in\sections{\Compact}{\cospinorbundle}$ while $\spinorpropagator{}{\genericmass}$, $\cospinorpropagator{}{\genericmass}$ are the causal propagator for $\Dirac{\genericmass}$ and $\Dirac{\genericmass}^\star$ respectively (see section \ref{space of solutions}).
A characterization of the quasi-free states on $\algebraF{\genericmass}$ was obtained by H. Araki in \cite{araki}: 

\begin{lem}\label{Araki's Theorem}
Let $Q$ be a linear bounded operator on $\Hilbert{\genericmass}\doteq\mathcal{H}_{(s),\beta}\oplus\mathcal{H}_{(c),\beta}$ satisfying:
\begin{gather}\label{Araki's conditions on S}
0\leq Q=Q^*\leq 1,\qquad
Q+\involution Q\involution = I_{\Hilbert{\genericmass}}.
\end{gather}
Then
\begin{gather*}\label{Araki's conditions on S I}
\state_2(\involution f,g)=\scalarproduct{f}{Qg}\qquad\forall f,g\in\labelspace{\genericmass},
\end{gather*}
defines a quasi-free state on $\algebraF{\genericmass}$. Conversely, for every quasi-free state on $\algebraF{\genericmass}$ there exists a bounded linear operator  $Q$ on $\Hilbert{\genericmass}$ fulfilling  \eqref{Araki's conditions on S}.
\end{lem}

It is widely accepted that, among all possible states, the physical ones are required to satisfy the Hadamard condition  \cite{Radzikowski, BFK, ChrisVerch}, namely a state $\state$ on $\algebra{\genericmass}$ is said to be of  \textbf{Hadamard form} if and only if the two-point bidistribution $\widetilde{\state_2}$ satisfies:
\begin{gather*}\label{Microlocal Spectrum Condition}
\WaveFrontSet{\widetilde{\state_2}}=\{(x,y,k_x,-k_y)\in\cotangentbundle^2\backslash\{0\}|\ (x,k_x)\sim(y,k_y),\ k_x\rhd 0\},
\end{gather*}
where $(x,k_x)\sim(y,k_y)$ means that there is a null geodesic $\gamma$ connecting $x$ to $y$, such that $k_x$ is cotangent to $\gamma$ at $x$ and $k_y$ is the coparallel transport along $\gamma$ of $k_x$ from $x$ to $y$. In addition, $k_x\rhd 0$ selects future pointing covectors. 
Since we deal with vector-valued distributions, the standard convention for the wavefront set is to take the union of the wavefront set of its components in an arbitrary but fixed local frame. It turns out that this definition does not depend on the chosen local frame, see \cite{SahlmannVerch} .
For completeness and historical reasons, we remark that this condition can also be characterized in any geodesically convex neighborhood. We refer to \cite{Fulling1, Fulling2} for more details.

\section{Fermionic Projector}\label{Fermionic Projector state}

The aim of the section is to focus on the construction of the so-called \textbf{Fermionic Projector}. To achieve our goal, we refer to an earlier work due to Finster and Reintjes \cite{Felix2}, though we also benefit from \cite{simo}.

Let $\alpha\in\mathbb{R}\setminus \{0\}$ and $\massinterval\subseteq\real\setminus\{0\}$ be an open interval containing $\massofinterest$. Consider an arbitrary but fixed Cauchy surface $\Cauchysurface$ and let $\Psi^0\in C^\infty_{\Compact}(\Cauchysurface\times\massinterval,\complex^4)$ be a smooth compactly supported function on $\Cauchysurface\times\massinterval$. Solving for every $\genericmass \in I_\alpha$ the Cauchy problem
\begin{align*}
\Dirac{\genericmass}\Psi(x,\beta) =0 \qquad \Psi(x,\beta) \big|_\Sigma = \Psi^0 \, ,
\end{align*}
we obtain a family $\Psi$ of solutions of the Dirac equation for a variable mass parameter $\beta\in I_\alpha$ in the class $C^\infty_{\SpacelikeCompact,\Compact}(\spacetime\times I_\alpha, \complex^4)$.
We endow the $\complex$-vector space built out of such families of solutions with the following scalar product
\begin{gather}\label{Direct sum scalar product}
\directsumscalarproduct{\Psi_1}{\Psi_2}\doteq\int_{\massinterval}\spinorscalarproduct{\Psi_{1,\genericmass}}{\Psi_{2,\genericmass}}\textrm{d}\genericmass,
\end{gather}
where $\textrm{d}\genericmass$ is the Lebesgue measure and $\spinorscalarproduct{\cdot}{\cdot}$ is the scalar product defined in \eqref{spinor scalar product}, which involves integration over $\Cauchysurface$.
Taking the completion of such space, we obtain the Hilbert space $\HilbertDirectSum$: It contains measurable families $\Psi\doteq (\Psi_\beta)_{\beta \in I_\alpha}$ such that $\Psi_\genericmass\in\Hilbertspinor{\genericmass}$ for almost all $\genericmass\in\massinterval$ and $\Psi_{\genericmass}|_{\Cauchysurface}$ is square integrable on any Cauchy surface $\Cauchysurface$.
We denote with $\SequenceSpinorm{\cdot}$ the norm of $\HilbertDirectSum$.\\
In $\HilbertDirectSum$ we can distinguish two different dense subspaces.
The first one, denoted $\soldirectsum$, is the collection of of families $\widetilde{\Psi}\doteq (\widetilde{\Psi}_\beta)_{\beta \in I_\alpha}$ such that $\widetilde{\Psi}_\genericmass\in\sol{\genericmass}$ for almost all $\genericmass\in\massinterval$.
The second one, denoted $\Hilbertspinor{\infty}\doteq C^\infty_{\SpacelikeCompact,\Compact}(\spacetime\times\massinterval,\complex^4)\cap\HilbertDirectSum$, is built out of families $\Psi\doteq (\Psi_\beta)_{\beta \in I_\alpha}$ which are compact in the mass parameter.
On $\HilbertDirectSum$, we introduce two self-adjoint operators: The first one multiplies by $\genericmass$
\begin{gather*}
\MassMultiplicativeOperator:\HilbertDirectSum\to\HilbertDirectSum \qquad (\MassMultiplicativeOperator\Psi)_\genericmass \doteq \genericmass \Psi_\genericmass\, .
\end{gather*}
while the second one, denoted as \textbf{smearing operator}, integrates over $\genericmass$
\begin{gather}\label{smearing operator}
\SmearingMassOperator\colon\Hilbertspinor{\infty}\to\sections{\SpacelikeCompact}{\spinorbundle}\qquad
\SmearingMassOperator\Psi(x)\doteq\int_{\massinterval}\Psi_{\genericmass}(x)\textrm{d}\genericmass.
\end{gather}
Notice that, even if $\Psi\in\soldirectsum$, $\SmearingMassOperator\Psi\not\in \sol{\genericmass}$ regardless the choice of  $\genericmass$.
The smearing operator plays a central role since it allows to construct the sequilinear form
\begin{gather}\label{Mass Smeared Inner Product}
\MassSmearedinnerproductOperatorNotation\colon\domain{\MassSmearedinnerproductOperatorNotation}\to\complex,\qquad
\MassSmearedinnerproductOperatorNotation(\Psi_1,\Psi_2)\doteq
\MassSmearedinnerproduct{\Psi_1}{\Psi_2}\doteq
\bracket{\SmearingMassOperator\Psi_1}{\adjunction{}\SmearingMassOperator\Psi_2},
\end{gather}
with $\domain{\MassSmearedinnerproductOperatorNotation}\doteq\{(\Psi_1,\Psi_2) \in \HilbertDirectSum \times \HilbertDirectSum |\ \MassSmearedinnerproduct{\Psi_1}{\Psi_2}\textrm{exists finite}\}\subseteq\Hilbertspinor{\infty}\times\Hilbertspinor{\infty}$, where $\bracket{}{}$ is defined in \eqref{bracket}. 

\begin{defn}[\cite{Felix2}]\label{Definition of WMOP}
The Dirac operator $\Dirac{\alpha}$ on $(\spacetime,g)$ has the \textbf{weak mass oscillation property} (WMOP) in the interval $\massinterval$ with domain $\Hilbertspinor{\infty}$ if for any $\Psi_1\in\Hilbertspinor{\infty}$ 
\begin{itemize}
\item[1)] $\exists\, C=C(\Psi_1) >0 $ such that
\begin{gather}\label{WMOP}
|\MassSmearedinnerproduct{\Psi_1}{\Psi_2}|\leq C(\Psi_1) \SequenceSpinorm{\Psi_2}\qquad\forall\Psi_2\in\Hilbertspinor{\infty}
\end{gather}
\item[2)] it holds
\begin{gather}\label{Commutation with MassMultiplicativeOperator}
\innerproduct{\SmearingMassOperator\MassMultiplicativeOperator\Psi_1}{\SmearingMassOperator\Psi_2}=
\innerproduct{\SmearingMassOperator\Psi_1}{\SmearingMassOperator\MassMultiplicativeOperator\Psi_2}
\qquad\forall\Psi_1,\Psi_2\in\Hilbertspinor{\infty}.
\end{gather} 
\end{itemize} 
\end{defn} 
\noindent Whenever $\MassSmearedinnerproductOperatorNotation$ satisfies the WMOP, we can represent it as the linear symmetric operator  $\FermionicOperator\colon\Hilbertspinor{\infty}\to\HilbertDirectSum$, namely
\begin{equation}\label{Defining property of the Fermionic Operator}
\MassSmearedinnerproductOperatorNotation(\Psi_1,\Psi_2)=
\directsumscalarproduct{\FermionicOperator\Psi_1}{\Psi_2} \, .
\end{equation}
We stress that the integral on the left hand side of equation \eqref{Defining property of the Fermionic Operator} is over $\spacetime$ while the one on the right hand side is taken on an arbitrary but fixed Cauchy surface $\Cauchysurface$.
We refer to $\FermionicOperator$ as the \textbf{Fermionic Signature operator}. 
We want to underline that up to this stage, we are not able to define a self-adjoint operator on $\Hilbertspinor{\massofinterest}$. Therefore a stronger condition is needed.

\begin{defn}[\cite{Felix2}]\label{Definition of SMOP}
The Dirac operator $\Dirac{\alpha}$ on $(\spacetime,g)$ possesses the \textbf{strong mass oscillation property} (SMOP) in the interval $\massinterval$ with domain $\Hilbertspinor{\infty}$ if there exists a constant $C>0$ such that
\begin{gather}\label{SMOP}
|\MassSmearedinnerproduct{\Psi_1}{\Psi_2}|\leq C\int_{\massinterval}\spinorm{\Psi_{1,\genericmass}}\spinorm{\Psi_{2,\genericmass}}\textrm{d}\genericmass
\qquad\forall\Psi_1,\Psi_2\in\Hilbertspinor{\infty}.
\end{gather}
\end{defn}
\noindent The SMOP plays a crucial role, since it is equivalent to the following two conditions:
\begin{itemize}
\item[(i)] for any $\Psi_1$, $\Psi_2\in\Hilbertspinor{\infty}$ it holds (\ref{Commutation with MassMultiplicativeOperator}) and
\begin{gather}\label{reduced SMOP}
|\MassSmearedinnerproduct{\Psi_1}{\Psi_2}|\leq C\SequenceSpinorm{\Psi_1}\SequenceSpinorm{\Psi_2};
\end{gather}
\item[(ii)] there exists a family $(\FermionicProjectionOperator{\genericmass})_{\genericmass\in\massinterval}$, where each $\FermionicProjectionOperator{\genericmass}$ acts on $\Hilbertspinor{\genericmass}$ as a self-adjoint linear bounded operator, such that
\begin{subequations}
\begin{gather}
%\label{boundness of associated family related to the Fermionic Operator}
\label{weak continuity of associated family related to the Fermionic Operator}
\sup_{\genericmass\in\massinterval}\|\FermionicProjectionOperator{\genericmass}\|<+\infty,\quad
\genericmass\mapsto\spinorscalarproduct{\Psi_{1,\genericmass}}{\FermionicProjectionOperator{\genericmass}\Psi_{2,\genericmass}}\textrm{ is continuous }
\forall\Psi_1,\Psi_2\in\Hilbertspinor{\infty}\\
\label{diagonalisation due to SMOP}
%\MassSmearedinnerproduct{\Psi_1}{\Psi_2}=
\directsumscalarproduct{\FermionicOperator\Psi_1}{\Psi_2}=
\int_{\massinterval}\spinorscalarproduct{\FermionicProjectionOperator{\genericmass}\Psi_{1,\genericmass}}{\Psi_{2,\genericmass}}\dvol \genericmass
\qquad\forall\Psi_1,\Psi_2\in \HilbertDirectSum.
\end{gather}
\end{subequations}
\end{itemize}
For the technical details, we refer to Theorem 4.2 and to Proposition 4.4 in \cite{Felix2}.
As a direct consequence of \eqref{reduced SMOP}, the SMOP implies the WMOP. Moreover, in formula (\ref{diagonalisation due to SMOP}) a bounded, symmetric operator $\FermionicProjectionOperator{\massofinterest}$ on $\Hilbertspinor{\massofinterest}$ is defined.
%Moreover, formula (\ref{diagonalisation due to SMOP}) allows to go back to the original Hilbert space $\Hilbertspinor{\massofinterest}$ by selecting the operator $\FermionicProjectionOperator{\massofinterest}$: For that it is also important the continuity property (\ref{weak continuity of associated family related to the Fermionic Operator}). 
Using spectral calculus, we can define  the \textbf{Fermionic Projector} to be 
$$\projector_{\textrm{FP}}\doteq\chi(\FermionicProjectionOperator{\massofinterest}) = \int_{\sigma(\FermionicProjectionOperator{\massofinterest})}\chi(\lambda)\, \textrm{d}\mu_\lambda\, : \Hilbertspinor{\massofinterest} \to \Hilbertspinor{\massofinterest}\, ,$$
where $\chi$ is the Heaviside step function, while $\textrm{d}\mu_\lambda$ is spectral measure associated to $\FermionicProjectionOperator{\massofinterest}$. 

\begin{rem}
Out of the Fermionic Projector, we can construct a quasi-free state on the algebra of Dirac fields. This can be done by defining the operator $Q\doteq \projector_{\textrm{FP}}\oplus(I_{\Hilbertcospinor{\genericmass}}-\adjunction{}\projector_{\textrm{FP}}\inverse{\adjunction{}})$ on $\Hilbert{\genericmass}$. By direct inspection it satisfies  both conditions in (\ref{Araki's conditions on S}); hence, applying Lemma \ref{Araki's Theorem}, we obtain a quasi-free state denoted by $\FermionicProjectorState$. We refer to it as \textbf{FP state}.\\ On a generic spacetime, this construction can be thought as a generalization of the frequency splitting for the Hamiltonian of the theory and on Minkowski spacetime, the usual vacuum state is recovered.
In general, it is not clear if the state $\FermionicProjectorState$ will be Hadamard or not. However, there are already explicit cases \cite{simo, Felix5} in which the construction is known to enjoy this property. 
\end{rem}

\section{Modified Fermionic Projector state}\label{Modified Fermionic Projector state}

From the mathematical point of view, the construction of $\FermionicProjectorState$ depends strongly on the SMOP, which must be checked case-by-case on each spacetime. It would be desirable to individuate a weaker requirement for implementing the whole procedure. 
In the following we show that this indeed possible, though we have to rely on a non-canonical construction. This is our main result.
  The key idea is to avoid the SMOP by defining a continuous immersion $\Hilbertspinor{\massofinterest}\hookrightarrow\HilbertDirectSum$ through a suitable bounded map $\SequenceMoller$. Composing the map $\SmearingMassOperator$ described in \eqref{smearing operator} with $\SequenceMoller$, leads to a modified version both of the SMOP and of the WMOP.
In particular, the modified MOPs are now formulated directly on $\Hilbertspinor{\massofinterest}$.
Thus it will be enough to check the latter propriety to define a pure quasi-free state on $\algebraF{\massofinterest}$.

\subsection{Sequence intertwining operator.}
In this section we introduce the embeddings $\SequenceMoller:\Hilbertspinor{\massofinterest}\to\HilbertDirectSum$
which are nothing but a direct sum of ``M\o ller type'' maps.
These are well known in the literature and they allow to intertwine the dynamics of two Green hyperbolic operators differing by a smooth potential.
The M\o ller map was used by Peierls \cite{Peierls} as a general procedure to define the Poisson brackets for the algebra of observables.
Results on the existence of M\o ller operators can be found in \cite{Brunetti e Fredenhagen, Dutsch Fredenhagen, nico1} and references therein. 
With this in mind, we introduce first the intertwining map $\Moller{\genericmass,\massofinterest}:\Hilbertspinor{\massofinterest}\to\Hilbertspinor{\genericmass}$, $\genericmass\in\massinterval$, eventually describing its direct sum as a map $\SequenceMoller:\Hilbertspinor{\massofinterest}\to\HilbertDirectSum$.
Consider two Cauchy surfaces $\Cauchysurface_\pm$ such that $\Cauchysurface_+$ lies in the future of $\Cauchysurface_-$.
Let $\rho^+\in C^\infty(\real)$ be a non decreasing function such that $\rho^+|_{J^+(\Cauchysurface_+)}=1$ and $\rho^+|_{J^-(\Cauchysurface_-)}=0$. We introduce $\rho^-=1-\rho^+$.
For any $\psi\in\sol{\massofinterest}$, we define $\Moller{\genericmass,\massofinterest}\psi\in\sol{\genericmass}$ as follows.
First consider the unitary operator which maps the Cauchy data $\psi|_{\Cauchysurface_-}$ to the corresponding Cauchy data on $\Cauchysurface_+$ by evolving it via the dynamics ruled by the Dirac equation with mass $\interpolatingmass(\genericmass,\massofinterest)=\genericmass\rho^++\massofinterest\rho^-$. Secondly define $\Moller{\genericmass,\massofinterest}\psi$ as the solution of the Dirac equation with mass $\genericmass$ and Cauchy data provided by those previously obtained on $\Cauchysurface_+$.
The whole procedure can be described explicitly as the composition of the following maps:
\begin{gather*}
\Moller{\interpolatingmass,\massofinterest}^+=
\textrm{Id}-
\spinorpropagator{+}{\interpolatingmass}(\genericmass-\massofinterest)\rho^+,\qquad
\Moller{\genericmass,\interpolatingmass}^-=
\textrm{Id}-
\spinorpropagator{-}{\genericmass}(\massofinterest-\genericmass)\rho^-,\qquad
\Moller{\genericmass,\massofinterest}=
\Moller{\genericmass,\interpolatingmass}^-\circ\Moller{\interpolatingmass,\massofinterest}^+,
\end{gather*}
where $E^+_\interpolatingmass$ denotes the advanced propagator for the Dirac equation with mass $\interpolatingmass$.
Note that $(\genericmass-\massofinterest)\rho^+$, $(\massofinterest-\genericmass)\rho^-$, are past- and future-compact respectively; thus the composition with $\propagator{+}{\interpolatingmass}$, $\propagator{-}{\genericmass}$ is well defined.\\
The map $\Moller{\massofinterest,\genericmass}$ is thus densely defined as a linear unitary map from $\sol{\massofinterest}$ to $\sol{\genericmass}$, whose extension will be denoted again by $\Moller{\massofinterest,\genericmass}$.
We define $\SequenceMoller:\Hilbertspinor{\massofinterest}\to\HilbertDirectSum$ by
\begin{gather}\label{Sequence Moller operator}
\SequenceMoller\psi(\genericmass)\doteq
\Moller{\genericmass,\massofinterest}\psi\in\Hilbertspinor{\genericmass}
\qquad
\forall\genericmass\in\massinterval,\quad
\forall\psi\in\HilbertDirectSum.
\end{gather}
In view of the continuous dependence of the solutions of the Dirac equation on the mass parameter (see \cite{Taylor}), the function $\genericmass\mapsto\Moller{\genericmass,\massofinterest}\psi$ is integrable in the sense required by the definition of $\HilbertDirectSum$: Furthermore
\begin{gather*}\label{bound for the Sequence Moller}
\norm{\SequenceMoller\psi}{\massinterval}^2=
\int_{\massinterval}\spinorm{\Moller{\genericmass,\massofinterest}\psi}^2\textrm{d}\genericmass=
|\massinterval|\spinorm{\psi}^2<\infty,
\end{gather*}
thus proving that $\SequenceMoller\psi\in\HilbertDirectSum$ and that $\SequenceMoller$ is an almost isometric linear bounded operator from $\Hilbert{\massofinterest}$ to $\HilbertDirectSum$, with $\norm{\SequenceMoller}{}=\sqrt{|\massinterval|}$.
We refer to $\SequenceMoller$ as the \textbf{sequence intertwining operator}. 

\subsection{Modified Mass Oscillation properties.} 
The sequence intertwining operator \eqref{Sequence Moller operator} allows formulating  the mass oscillation properties directly on $\Hilbertspinor{\massofinterest}$.
Notice that, for any $\psi\in\Hilbertspinor{\massofinterest}$, the element $\SequenceMoller\psi$ lies in $\soldirectsum$ but a priori not in $\Hilbertspinor{\infty}$ since we have no control on the support properties of $\SequenceMoller\psi$ as a function of $\genericmass$.
Thus, in order to make contact with the Definitions \ref{Definition of WMOP}  and \ref{Definition of SMOP}, we localize $\SequenceMoller\psi$ with an arbitrary smooth, compactly supported function $\MassFunction{}\in C^\infty_{\Compact}(\massinterval)$. 

\begin{defn}\label{definizione di mWMOP e mSMOP}
The Dirac operator $\Dirac{\massofinterest}$ on $(\spacetime,g)$ has the \textbf{modified weak mass oscillation property} (mWMOP) in the interval $\massinterval$ with domain $\sol{\massofinterest}$ if, for any $\MassFunction{}\in C^\infty_{\Compact}(\massinterval)$ and $\psi_1\in\sol{\massofinterest}$, there exists a constant $C(\MassFunction{},\psi_1) >0 $ such that
\begin{gather}\label{mWMOP}
|\MassSmearedinnerproduct{\MassFunction{}\SequenceMoller\psi_1}{\MassFunction{}\SequenceMoller\psi_2}|\leq C(\MassFunction{},\psi_1)\spinorm{\psi_2}\qquad\forall\psi_2\in\sol{\massofinterest}.
\end{gather}
Similarly, $\Dirac{\alpha}$ on $(\spacetime,g)$ has the \textbf{modified strong mass oscillation property} (mSMOP) in the interval $\massinterval$ with domain $\sol{\massofinterest}$ if for any $\MassFunction{}\in C^\infty_{\Compact}(\massinterval)$ there exists a constant $C(\MassFunction{})>0$ such that
\begin{gather}\label{mSMOP}
|\MassSmearedinnerproduct{\MassFunction{}\SequenceMoller\psi_1}{\MassFunction{}\SequenceMoller\psi_2}|\leq C(\MassFunction{})\spinorm{\psi_1}\spinorm{\psi_2}\qquad
\forall\psi_1,\psi_2\in\sol{\massofinterest}.
\end{gather}
\end{defn}
\begin{rem}\label{proposition on relations between MOPs and mMOps}
We want to stress that comparing Definition \ref{definizione di mWMOP e mSMOP} with Definitions \ref{Definition of WMOP} and \ref{Definition of SMOP}, we avoid the commutation property (\ref{Commutation with MassMultiplicativeOperator}). This is a requirement used to show the equivalence between the different formulations of the SMOP (\ref{SMOP}-\ref{reduced SMOP}) and it plays no role in our construction.
\end{rem} 

We  describe now the relation between the new modified MOP and the properties introduced in Definition \ref{Definition of WMOP} and \ref{Definition of SMOP}.
	\begin{cor}
The modified mass oscillation properties are weaker requirements than the mass oscillation property, namely the following diagram holds
\begin{displaymath}
\xymatrix{
\textrm{SMOP}\ar@{=>}[d]&\ar@{<=}[l]\ar@{=>}[d]\textrm{WMOP}\\
\textrm{mSMOP}&\ar@{<=}[l]\textrm{mWMOP} \,.
}
\end{displaymath}
\end{cor}
\begin{proof}
Let be $\Psi_1, \Psi_2 \in \cH_{(s),\infty}$ and let $\psi_1,\psi_2\in\sol{\massofinterest}$. 
The horizontal arrows descend taking $C(\Psi_1) = C||\Psi_1||_{(s),I_\alpha}$ in \eqref{WMOP} and $C(\MassFunction{},\psi_1)= C(\MassFunction{})\spinorm{\psi_1}$ in \eqref{mWMOP}.
In order to show that SMOP$\Rightarrow$mSMOP (and similarly for the weak properties), it is enough to substitute $\Psi_1=\MassFunction{}\SequenceMoller\psi_1$, $\Psi_2=\MassFunction{}\SequenceMoller\psi_2$ in (\ref{reduced SMOP}) and to use\footnote{With a little more effort, one may prove that SMOP implies that, there is a constant $C>0$ such that, for any $\psi,\phi\in\sol{\massofinterest}$, it holds $|\MassSmearedinnerproduct{\SequenceMoller\psi}{\SequenceMoller\phi}|\leq C\spinorm{\psi}{}\spinorm{\phi}{}$. Since this alternative modified SMOP plays no role in the following, we will stick to the ``localized'' definition, which allows a more direct comparison between MOPs and mMOPs.} $\norm{\MassFunction{}\SequenceMoller\psi}{\massinterval}=\|\MassFunction{}\|_{L^2(\massinterval)}\spinorm{\psi}{}$.
This proves that the m-WMOP is a proper weaker requirement than both the SMOP and the WMOP.
\end{proof}

\subsection{New classes of Fermionic Projectors.} 
In this section, we describe how to build a quasi-free state from the modified MOPs.
For sake of completeness, we first discuss the case in which the mSMOP holds true.
Eventually, we focus on the case where  only the mWMOP holds true.
We stress that the main point here is to provide a construction of a state which also holds in cases where the method of \cite{Felix2} cannot be applied immediately. 

\begin{thm}\label{Teo:construction of mFP in mSMOP hp}
If the mSMOP holds true, then, for any $\MassFunction{}\in C^\infty_\Compact(\massinterval)$, there exists a unique self-adjoint operator $\ModifiedFermionicProjectionOperator{\massofinterest,\MassFunction{}}:\Hilbertspinor{\massofinterest}\to\Hilbertspinor{\massofinterest}$, henceforth called \textbf{modified Fermionic Signature operator}, defined by
\begin{gather}\label{modified Fermionic operator}
\scalarproduct{\ModifiedFermionicProjectionOperator{\massofinterest,\MassFunction{}}\psi_1}{\psi_2}=
\MassSmearedinnerproduct{\MassFunction{}\SequenceMoller\psi_1}{\MassFunction{}\SequenceMoller\psi_2}\qquad
\forall\psi_1\in\sol{\massofinterest},
\forall\psi_2\in\Hilbertspinor{\massofinterest}.
\end{gather}
The spectral decomposition of $\ModifiedFermionicProjectionOperator{\massofinterest,\MassFunction{}}$ yields a spectral projector 
\begin{equation}\label{FPs}
\projector_{\textrm{mFP}}=\chi(\ModifiedFermionicProjectionOperator{\massofinterest,\MassFunction{}}):\Hilbertspinor{\massofinterest}\to\Hilbertspinor{\massofinterest}
\end{equation}
and hence a quasi-free state $\omega_{\textrm{mFP}}:\algebraF{\genericmass} \to \complex$.
\end{thm}
\begin{proof}
If the mSMOP holds true then, for any $\MassFunction{}\in C^\infty_{\Compact}(\massinterval)$, by (\ref{mSMOP}) we can define a linear bounded operator $\ModifiedFermionicProjectionOperator{\massofinterest,\MassFunction{}}:\Hilbertspinor{\massofinterest}\to\Hilbertspinor{\massofinterest}$ defined via Riesz Theorem.
Indeed (\ref{mSMOP}) assures that, for any $\psi_1\in\sol{\massofinterest}$, $\psi_2\mapsto\MassSmearedinnerproduct{\MassFunction{}\SequenceMoller\psi_1}{\MassFunction{}\SequenceMoller\psi_2}$ is a densely defined continuous linear functional.
After extension on $\Hilbertspinor{\massofinterest}$ we can apply Riesz Theorem to obtain
\begin{gather*}
\scalarproduct{\ModifiedFermionicProjectionOperator{\massofinterest,\MassFunction{}}\psi_1}{\psi_2}=
\MassSmearedinnerproduct{\MassFunction{}\SequenceMoller\psi_1}{\MassFunction{}\SequenceMoller\psi_2}\qquad
\forall\psi_1\in\sol{\massofinterest},
\forall\psi_2\in\Hilbertspinor{\massofinterest},
\end{gather*}
where we already made explicit the linear dependence on $\psi_1$ of the element $\ModifiedFermionicProjectionOperator{\massofinterest,\MassFunction{}}\psi_1\in\Hilbertspinor{\massofinterest}$.
Thus, we have found the modified Fermionic Signature operator, namely a linear map $\ModifiedFermionicProjectionOperator{\massofinterest,\MassFunction{}}:\sol{\massofinterest}\to\Hilbertspinor{\massofinterest}$, which is also symmetric on account of (\ref{modified Fermionic operator}).
Notice that this procedure only makes use of the bound on $\psi_2$: Hence it is valid also in the case where the mWMOP holds true (see (\ref{mWMOP})).
In the case of the mSMOP, we can also conclude that $\ModifiedFermionicProjectionOperator{\massofinterest,\MassFunction{}}$ is bounded, actually $\norm{\ModifiedFermionicProjectionOperator{\massofinterest,\MassFunction{}}}{}\leq C(\MassFunction{})$.
We have thus a self-adjoint operator $\ModifiedFermionicProjectionOperator{\massofinterest,\MassFunction{}}:\Hilbertspinor{\massofinterest}\to\Hilbertspinor{\massofinterest}$, whose spectral decomposition allows to define a projector $\projector_{\textrm{mFP}}=\chi(\ModifiedFermionicProjectionOperator{\massofinterest,\MassFunction{}})$.
From it we can construct a state, $\omega_{\textrm{mFP}}$ by applying Lemma \ref{Araki's Theorem}, once we defined 
\begin{gather*}
Q\doteq \projector_{\textrm{mFP}}\oplus(I_{\Hilbertcospinor{\genericmass}}-\adjunction{}\projector_{\textrm{mFP}}\inverse{\adjunction{}}),
\end{gather*}
where $\adjunction{}$ as been introduced in \eqref{adjunction map} .
This completes the construction of the modified Fermionic Projector state in the case of the mSMOP.
\end{proof} 

We deal now with the case in which the mWMOP holds true but not the mSMOP. 

\begin{thm}\label{Teo:construction of mFP in mWMOP hp}
If the mWMOP holds true, then there exists a densely defined, symmetric operator $\ModifiedFermionicProjectionOperator{\massofinterest,\MassFunction{}}:\Hilbertspinor{\massofinterest}\to\Hilbertspinor{\massofinterest}$. Moreover we obtain a quasi-free state $\omega_{\text{mFP}}:\algebraF{\genericmass} \to \complex$ out of the operator $\projector_{\textrm{mFP}}:\Hilbertspinor{\massofinterest}\to\Hilbertspinor{\massofinterest}$ defined by
\begin{equation}\label{FPw}
\projector_{\textrm{mFP}}=\frac{1}{2} \int_{\sigma(\ModifiedFermionicProjectionOperator{\massofinterest}^2)} \rho^{-\frac{1}{2}}( \ModifiedFermionicProjectionOperator{\massofinterest} + \rho^{\frac{1}{2}}) \,\textrm{d}\mu_\rho
\end{equation}
where $\textrm{d}\mu$ is the spectral measure associated to $\ModifiedFermionicProjectionOperator{\massofinterest}$.
\end{thm} 
\begin{proof}
Following the first part of the proof of Theorem \ref{Teo:construction of mFP in mSMOP hp}, the mWMOP in combination with the Riesz representation Theorem allow us to introduce, for all choices of $\MassFunction{}\in C^\infty_{\Compact}(\massinterval)$, a densely defined symmetric linear operator $\ModifiedFermionicProjectionOperator{\massofinterest,\MassFunction{}}$
\begin{gather*}
\scalarproduct{\ModifiedFermionicProjectionOperator{\massofinterest,\MassFunction{}}\psi_1}{\psi_2}=
\MassSmearedinnerproduct{\MassFunction{}\SequenceMoller\psi_1}{\MassFunction{}\SequenceMoller\psi_2}\qquad
\forall\psi_1\in\sol{\massofinterest},
\forall\psi_2\in\Hilbertspinor{\massofinterest}.
\end{gather*}
However, this time we have a priori no boundedness condition on $\ModifiedFermionicProjectionOperator{\massofinterest,\MassFunction{}}$.
Nevertheless, we can still use techniques similar to the ones used in \cite{Felix2} (cfr. Section 3.2) to obtain the projector $\projector_{\textrm{mFP}}$.
First we use the Friederich method to construct a self-adjoint extension for $\ModifiedFermionicProjectionOperator{\massofinterest}^2$ (see \cite{Lax} for more details); subsequently we define a spectral projector of $\ModifiedFermionicProjectionOperator{\massofinterest}$ as
\begin{gather*}
\projector_{\textrm{mFP}}\doteq\chi(\ModifiedFermionicProjectionOperator{\massofinterest}) \doteq \frac{1}{2 \sqrt{\ModifiedFermionicProjectionOperator{\massofinterest}^2}} \left( \ModifiedFermionicProjectionOperator{\massofinterest} + \sqrt{\ModifiedFermionicProjectionOperator{\massofinterest}^2} \right) = \frac{1}{2} \int_{\sigma(\ModifiedFermionicProjectionOperator{\massofinterest}^2)} \rho^{-\frac{1}{2}}( \ModifiedFermionicProjectionOperator{\massofinterest} + \rho^{\frac{1}{2}}) \,\dvol\mu_\rho,
\end{gather*}% which again allows to conclude the construction of $\ModifiedFermionicProjectorState$.
where $\dvol\mu$ is the spectral measure associated to $\ModifiedFermionicProjectionOperator{\massofinterest}$. As before, we obtain the quasi-free state $\omega_{\text{mFP}}$ by introducing 
$Q\doteq \projector_{\textrm{FP}}\oplus(I_{\Hilbertcospinor{\genericmass}}-\involution\projector_{\textrm{FP}}\involution)$ and by applying Lemma \ref{Araki's Theorem}.
\end{proof}
\begin{rem}
Whenever the modified Fermionic Signature operator  $\ModifiedFermionicProjectionOperator{\massofinterest,\MassFunction{}}$ is also self-adjoint, formula \eqref{FPw} is reduced to \eqref{FPs}. 
\end{rem}

To conclude this section we analyze of the Hadamard condition for the class of states $\omega_{\textrm{mFP}}$.
A general treatment would lead to a case-by-case analysis because it is a priori not clear which class of globally hyperbolic spacetimes $\mathscr{M}$ enjoys the modified strong/weak Mass Oscillation Properties.
Nevertheless, a positive answer can be obtained whenever the FP state $\omega_{\textrm{FP}}$ satisfies the Hadamard condition. This entails that our construction leads to a new class of Hadamard states.
\begin{thm}\label{Theorem: Hadamard property in a special case}
Let assume that the SMOP holds true on the interval $I_\alpha$ and that the states $\omega_{FP,\beta}$ constructed out of $\Pi_{\textrm{FP},\beta}\doteq\chi(\FermionicProjectionOperator{\beta})$ satisfy the Hadamard condition for all $\beta\in I_\alpha$.
Then the state $\omega_{\textrm{mFP}}$ constructed out of $\Pi_{\textrm{mFP}}\doteq\chi(\ModifiedFermionicProjectionOperator{\alpha,\mathfrak{m}})$ satisfies the Hadamard condition.
\end{thm}
\begin{proof}
First, we make connection with the construction of \cite{Felix2}, in particular we discuss the relation between the operators $\FermionicProjectionOperator{}$ and $\ModifiedFermionicProjectionOperator{\massofinterest,\MassFunction{}}$.
Let assume that the SMOP holds true.
For any $\psi_1,\psi_2\in\sol{\massofinterest}$ it holds
\begin{align*}
\scalarproduct{\ModifiedFermionicProjectionOperator{\massofinterest,\MassFunction{}}\psi_1}{\psi_2}=&
\MassSmearedinnerproduct{\MassFunction{}\SequenceMoller\psi_1}{\MassFunction{}\SequenceMoller\psi_2}\\ =& \directsumscalarproduct{\FermionicOperator\MassFunction{}\SequenceMoller\psi_1}{\MassFunction{}\SequenceMoller\psi_2}=
\scalarproduct{(\MassFunction{}\SequenceMoller)^\dagger\FermionicOperator\MassFunction{}\SequenceMoller\psi_1}{\psi_2}\,.
\end{align*}
Thus we find $\ModifiedFermionicProjectionOperator{\massofinterest,\MassFunction{}}=(\MassFunction{}\SequenceMoller)^\dagger\FermionicOperator\MassFunction{}\SequenceMoller$ or explicitly
\begin{gather*}
\ModifiedFermionicProjectionOperator{\massofinterest,\MassFunction{}}\psi_1=
\int_{\massinterval}\textrm{d}\genericmass
|\MassFunction{\genericmass}|^2\Moller{\genericmass,\massofinterest}^\dagger\FermionicProjectionOperator{\genericmass}\Moller{\genericmass,\massofinterest}\psi_1\,.
\end{gather*}
Since $R_{\beta,\alpha}$ is unitary one finds that
\begin{gather*}
\Pi_{\textrm{mFP}}\doteq
\chi(\ModifiedFermionicProjectionOperator{\massofinterest,\MassFunction{}})=
\int_{\massinterval}\textrm{d}\genericmass
|\MassFunction{\genericmass}|^2\Moller{\genericmass,\massofinterest}^\dagger\Pi_{\textrm{FP},\beta}\Moller{\genericmass,\massofinterest}\,,\quad
\Pi_{\textrm{FP},\beta}\doteq\chi(\FermionicProjectionOperator{\genericmass})\,.
\end{gather*}
Since by hypothesis $\Pi_{\textrm{FP},\beta}$ leads to a Hadamard state, one can use standard arguments on the propagation of singularities (see \cite[Theorem 4.1]{nico} or \cite[Theorem 5.2.1]{simoPhD} to infer that $\Moller{\genericmass,\massofinterest}^\dagger\Pi_{\textrm{FP},\beta}\Moller{\genericmass,\massofinterest}$ satisfies the Hadamard property.
Since the integration over $\beta$ does not affect the singular behavior of the state, we obtain that $\omega_{\textrm{mFP}}$ satisfies the Hadamard condition.
\end{proof}

\begin{rem}
\textit{(i)}
The hypothesis of Theorem \ref{Theorem: Hadamard property in a special case} are satisfied in Minkowski spacetime in the presence of a time-dependent external field \cite{simo} and of a plane electromagnetic wave \cite{Felix5}. 
\\
\textit{(ii)}
Variations on the construction of $\omega_{\textrm{mFP}}$ can be obtained by varying the spectral function $\chi$ in the definition of $\Pi_{mFP}$.
Indeed, any real-valued bounded function $f$ on the spectrum of $\ModifiedFermionicProjectionOperator{\alpha,\mathfrak{m}}$ would lead to a state $\omega_{\textrm{mFP},f}$ constructed out of $f(\ModifiedFermionicProjectionOperator{\alpha,\mathfrak{m}})$.
The arbitrariness of $f$ can be used in concrete examples to select states $\omega_{\textrm{mFP},f}$ satisfying the Hadamard condition, in the spirit of \cite{ChrisBenni}.
\\
\textit{(iii)}
By construction, our procedure depends on the chosen mass cut-off $\MassFunction{}\in C^\infty_{\Compact}(\massinterval)$.
Definition \ref{definizione di mWMOP e mSMOP} allows to perform the construction of $\omega_{\textrm{mFP}}$ for any choice of $\MassFunction{}$, but it does not fix any continuous dependence of this latter parameter.
\footnote{One could modify the definition of both mWMOP and mSMOP to avoid the mass cut off $\MassFunction{}$, but this would create additional difficulties in the comparison between the mMOPs and the original MOPs.}
Thus we do not expect, in general, to be able to perform a limit $\MassFunction{}\to 1$ which would remove the dependence on $\MassFunction{}$.\\
Finally, notice that for the whole construction of the Fermionic Signature operators $\FermionicProjectionOperator{\massofinterest}$ and $\ModifiedFermionicProjectionOperator{\massofinterest}$, we restricted ourself to the case of constant mass $\massofinterest$.
In the case of nonconstant $\massofinterest\in C^\infty(\spacetime,\real)$ we can still successfully apply the intertwining operator $\Moller{\genericmass,\massofinterest}$, but it is not immediately clear what should be the analogous of the space $\HilbertDirectSum$.
One may try to consider $\beta\in C^\infty(\spacetime,\real)\cap L^2(\spacetime,g)$, thus inducing a Gaussian measure on that space: The space $\HilbertDirectSum$ may than be defined as in Section \ref{Fermionic Projector state} with $\textrm{d}\genericmass$ substituted by the Gaussian measure.
\end{rem}

\section{Rindler spacetime and the mass oscillation properties}\label{Rindler counterexample} 

We describe now an explicit example to construct a modified Fermionic Projector state which could not have been dealt in \cite{Felix2}.
More precisely, we consider a spacetime where neither the SMOP nor the WMOP holds true, while the mWMOP does. 

Let $\Minkowski$ be the four dimensional Minkowski spacetime and let $\Rindler$ be the Rindler spacetime, defined as
\begin{equation*}
\Rindler \doteq\{(t,x,y,z) \in \Minkowski \ | \ |t| \leq x \} \, .
\end{equation*}
$\Rindler$ is globally hyperbolic spacetime and a foliation by smooth Cauchy hypersurfaces can be given, fixing hyperbolic coordinates $(t,x,y,z)=(r\sinh(s),r\cosh(s),y,z)$, as $\Cauchysurface_{s_0}=\{s=s_0\}$.
We prove now that, while the SMOP holds true on $\Minkowski$ (see \cite{Felix2,simo}), it does not in $\Rindler$.
The reason is the failure of (\ref{Commutation with MassMultiplicativeOperator}) on the latter space. Thus none of the mass oscillation properties introduced in the Definitions \ref{Definition of WMOP} and \ref{Definition of SMOP} can be satisfied.
Nevertheless, we show that the inequality (\ref{WMOP}) still holds true, thus implying the validity of the mWMOP. \\
First, we focus on \eqref{WMOP}: The proof follows the same path described in \cite{Felix2,simo} in the case of Minkowski spacetime.
This is related to the observation that a spacelike compact solution $\psi_\Rindler\in\sol{\massofinterest}_{\Rindler}$ of the Dirac equation in Rindler spacetime can be obtained by restriction of a spacelike compact solution $\psi_{\Minkowski}\in\sol{\massofinterest}_{\Minkowski}$ of the same equation on the whole Minkowski spacetime.
Moreover the norm of $\psi_{\Rindler}$ in the space $\overline{\sol{\massofinterest}}_{\Rindler}$ coincides with the norm of $\psi_{\Minkowski}$ on $\overline{\sol{\massofinterest}}_{\Minkowski}$, i.e. $\norm{\psi_{\Rindler}}{\Rindler}=\norm{\psi_{\Minkowski}}{\Minkowski}$. We can write $\psi_\Rindler=\bold{1}_\Rindler\psi_{\Minkowski}$, being $\bold{1}_\Rindler$ the characteristic function of $\Rindler$. \\

\begin{prop}
In Rindler spacetime, the weak and the strong mass oscillation properties do not hold, while the modified weak mass oscillation property \eqref{mWMOP} does.
\end{prop} 
\begin{proof}
We first prove that \eqref{WMOP} holds true on $\Rindler$.
Let $\Psi_{\Rindler},\Phi_{\Rindler}\in\Hilbertspinor{\infty,\Rindler}$ be two families of solutions of the Dirac equation on $\Rindler$ with spacelike compact support which are also compactly supported in the mass parameter, as described in Section \ref{Fermionic Projector state} (the index $_{\Rindler}$ stresses that the latter space is the one obtained regarding $\Rindler$ as the ambient space).
Let $\MassSmearedinnerproductOperatorNotation_\Rindler:\domain{\MassSmearedinnerproductOperatorNotation_\Rindler}\to\complex,$ be the sesquilinear form introduced relatively to Rindler space:
\begin{gather*}
\MassSmearedinnerproductOperatorNotation_\Rindler(\Psi_{\Rindler},\Phi_{\Rindler})=
\langle\SmearingMassOperator\Psi_{\Rindler}|\SmearingMassOperator\Phi_{\Rindler}\rangle_\Rindler=
\int_{\Rindler}\fiberpairing{\SmearingMassOperator\Psi_{\Rindler}}{A\SmearingMassOperator\Phi_{\Rindler}}\dvol\mu_h,
\end{gather*}
being $\dvol\mu_h$ the induce volume measure on $\Rindler$ and being
\begin{align*}
\domain{\MassSmearedinnerproductOperatorNotation_\Rindler}\doteq & \{(\Psi_\Rindler,\Phi_\Rindler) \in \Hilbertspinor{\infty,\Rindler} \times \Hilbertspinor{\infty,\Rindler}|\ \MassSmearedinnerproduct{\Psi_1}{\Psi_2}\textrm{exists finite}\}\\
&\subseteq\Hilbertspinor{\infty,\Rindler} \times \Hilbertspinor{\infty,\Rindler}.
\end{align*}
Exploiting the relation between $\Psi_{\Rindler}$ and $\Psi_{\Minkowski}$, we find
\begin{eqnarray*}
|\MassSmearedinnerproductOperatorNotation_\Rindler(\Psi_{\Rindler},\Phi_{\Rindler})|&=&
\left|\int_{\Minkowski}\fiberpairing{\SmearingMassOperator\Psi_{\Minkowski}}{A\bold{1}_{\Rindler}\SmearingMassOperator\Phi_{\Minkowski}}\dvol\mu_g\right|\\&=&\left|
\int_\real \int_{\Cauchysurface_t}\fiberpairing{\SmearingMassOperator\Psi_{\Minkowski}}{A(\gamma^0)^2\bold{1}_\Rindler\SmearingMassOperator\Phi_{\Minkowski}}|_{\Cauchysurface_t}\volumeformSigma\dvol t\right|\\ &=&
\left|\int_\real\scalarproduct{\SmearingMassOperator\Psi_{\Minkowski}|_{\Cauchysurface_t}}{\gamma^0\bold{1}_\Rindler\SmearingMassOperator\Phi_{\Minkowski}|_{\Cauchysurface_t}}_t\dvol t
\right|,
\end{eqnarray*}
where in the last equality we have fixed the foliation of $\Minkowski$ in terms of the Cauchy hypersurfaces $\Cauchysurface_t=\{t=\textrm{constant}\}$ (note that such hypersurfaces are not Cauchy hypersurfaces for $\Rindler$).
The scalar product $\scalarproduct{\cdot}{\cdot}_t$ is the (time dependent) scalar product on $L^2(\Cauchysurface_t)$, formally equal to (\ref{spinor scalar product}).
However, note that $\scalarproduct{\SmearingMassOperator\Psi_{\Minkowski}|_{\Cauchysurface_t}}{\gamma^0\bold{1}_\Rindler\SmearingMassOperator\Phi_{\Minkowski}|_{\Cauchysurface_t}}_t$ is not time independent since none of the functions involved is a solution of the Dirac equation.
Nevertheless, for any $t\in\real$, both $\SmearingMassOperator\Psi_{\Minkowski}|_{\Cauchysurface_t}$ and $\bold{1}_{\Rindler}\SmearingMassOperator\Phi_{\Minkowski}|_{\Cauchysurface_t}$ lie in $L^2(\Cauchysurface_t)$.
Thus, by applying the Schwarz and the H\"older inequalities
\begin{gather*}
|\MassSmearedinnerproductOperatorNotation_\Rindler(\Psi_{\Rindler},\Phi_{\Rindler})|\leq
\int_\real\norm{\SmearingMassOperator\Psi_{\Minkowski}|_{\Cauchysurface_t}}{t}\norm{\SmearingMassOperator\Phi_{\Minkowski}|_{\Cauchysurface_t}}{t}\dvol t.
\end{gather*}
As discussed in Lemma 3.1 in \cite{simo}, we can control the latter integrands with
\begin{gather*}
\norm{\SmearingMassOperator\Phi_{\Minkowski}|_{\Cauchysurface_t}}{t}\leq
\sqrt{|\massinterval|}\norm{\Phi_{\Minkowski}}{\massinterval}=
\sqrt{|\massinterval|}\norm{\Phi_{\Rindler}}{\massinterval},\qquad
\norm{\SmearingMassOperator\Psi_{\Minkowski}|_{\Cauchysurface_t}}{t}\leq\frac{C(\Psi_{\Minkowski})}{1+t^2},
\end{gather*}
where the constant $C(\Psi_{\Minkowski})$ depends on the spatial Sobolev norm of $\Psi_{\Minkowski}|_{\Cauchysurface_t}$.
We can conclude that
\begin{gather}\label{mwm}
|\MassSmearedinnerproductOperatorNotation_\Rindler(\Psi_{\Rindler},\Phi_{\Rindler})|\leq
c(\Psi_{\Minkowski})\norm{\Phi_{\Rindler}}{\massinterval}\,.
\end{gather}
Thus \eqref{WMOP} holds true and on account of Remark \ref{proposition on relations between MOPs and mMOps} the mWMOP follows immediately. \\
We prove now that condition \eqref{Commutation with MassMultiplicativeOperator} fails: Using $\MassMultiplicativeOperator\Psi_{\Rindler}(\genericmass)=\Dirac{\genericmass}\Psi_{\Rindler}(\genericmass)$, partial integration gives a boundary term
\begin{gather*}
\MassSmearedinnerproductOperatorNotation_{\Rindler}(\MassMultiplicativeOperator\Psi_{\Rindler},\Phi_{\Rindler})-
\MassSmearedinnerproductOperatorNotation_{\Rindler}(\Psi_{\Rindler},\MassMultiplicativeOperator\Phi_{\Rindler})=
\int_{\partial\Rindler}\fiberpairing{\SmearingMassOperator\Psi_{\Minkowski}}{\SmearingMassOperator\Phi_{\Minkowski}}|_{\partial\Rindler},
%\innerproduct{\SmearingMassOperator\Psi_{\Minkowski}}{\bold{1}_{\partial\Rindler}\SmearingMassOperator\Phi_{\Minkowski}} %_{\Minkowski}.
\end{gather*}
where $\partial\Rindler \doteq\{(t,x,y,z) \in \Minkowski \ | \ |t| = x \}$.
To show that this latter contribution does not vanish in general, we observe that
\begin{gather}\nonumber
\int_{\partial\Rindler}\fiberpairing{\SmearingMassOperator\Psi_{\Minkowski}}{\SmearingMassOperator\Phi_{\Minkowski}}|_{\partial\Rindler}=
i\int_{\massinterval}\dvol{\genericmass}
\int_{\massinterval}\dvol{\widetilde{\genericmass}}
\int_{\real^2}\dvol{x_\perp}
\int_0^{+\infty}\!\!\!\!\!\!\dvol{s}\left[
(\overline{\psi_\genericmass}(\gamma^1-\gamma^0)\phi_{\widetilde{\genericmass}})(s,s,x_\perp)\right.\\
\label{boundary terms in coordinates}
\left.+
(\overline{\psi_\genericmass}(\gamma^1+\gamma^0)\phi_{\widetilde{\genericmass}})(-s,s,x_\perp)
\right],
\end{gather}
where $\psi_\genericmass(t,x,x_\perp)=\Psi_\Minkowski(t,x,x_\perp,\genericmass)$ and $\phi_{\widetilde{\genericmass}}(t,x,x_\perp)=\Phi_\Minkowski(t,x,x_\perp,\widetilde{\genericmass})$.
Note that, since $\psi_\genericmass$, $\phi_{\widetilde{\genericmass}}$ are solutions on $\Rindler$, the integrand vanishes if $s\not\in (0,+\infty)$. Therefore we can extend the integral over the whole real line.
Without loss of generality, we choose $\psi_\genericmass$ and $\phi_{\widetilde{\genericmass}}$ as positive frequency solutions, namely 
\begin{equation}
\label{Fourierexp}
\begin{split}
\psi_\genericmass(t,x,x_\perp)&=
\int\dvol{k_\perp}\int\dvol{k}
c_+(k,k_\perp,\genericmass)e^{-i\omega t}e^{ikx}e^{ik_\perp x_\perp},\\
\phi_{\widetilde{\genericmass}}(t,x,x_\perp)&=
\int\dvol{k_\perp}\int\dvol{k}
\widetilde{c}_+(k,k_\perp,\widetilde{\genericmass})e^{-i\widetilde{\omega} t}e^{ikx}e^{ik_\perp x_\perp}
\end{split}
\end{equation}
with $\omega^2=k^2+|k_\perp|^2+\genericmass^2$ (resp. $\widetilde{\omega}^2=k^2+|k_\perp|^2+\widetilde{\genericmass}^2$) and $c_+$ (resp. $\widetilde{c}_+$) a suitable smooth function in $k$ which is compactly supported in the mass variable $\genericmass$ (resp. $\widetilde{\genericmass}$).
For simplicity, we also assume that $c_+$ and $\widetilde{c}_+$ are symmetric in the $k$ variable.
Later on, we will argue that the dependence on the mass parameter can be chosen in such a way that the integral in \eqref{boundary terms in coordinates} does not vanish.
Using the Fourier representation \eqref{Fourierexp} in \eqref{boundary terms in coordinates}, we find
\begin{gather}
\nonumber
\int_{\real^2}\dvol{x_\perp}\int_\real\dvol{s}\left[
(\overline{\psi_\genericmass}(\gamma^1-\gamma^0)\phi_{\widetilde{\genericmass}})(s,s,x_\perp)+
(\overline{\psi_\genericmass}(\gamma^1+\gamma^0)\phi_{\widetilde{\genericmass}})(-s,s,x_\perp)\right]\\
\nonumber
=
\int_{\real^2}\dvol{k_\perp}\int_\real\dvol{k}\int_\real\dvol{p}\left[
\delta(k_--p_-)\overline{c_+}(k,k_\perp,\genericmass)(\gamma^1-\gamma^0)\widetilde{c}_+(p,k_\perp,\widetilde{\genericmass})\right.\\
\label{Boundary terms in Fourier space}
\left.+
\delta(k_+-p_+)\overline{c_+}(k,k_\perp,\genericmass)(\gamma^1+\gamma^0)\widetilde{c}_+(p,k_\perp,\widetilde{\genericmass})
\right],
\end{gather}
where $k_\pm =\omega\pm k$, $p_\pm=\widetilde{\omega}\pm p$.
Changing variables and exploiting the symmetry of $c_+$, $\widetilde{c_+}$ in $k$ we are lead to
\begin{gather*}
\eqref{Boundary terms in Fourier space}=2
\int_0^{+\infty}\!\!\!\!\!\dvol{k}\ r(k,k_\perp,\genericmass)r(k,k_\perp,\widetilde{\genericmass})\overline{c_+}\left(\frac{k^2-\genericmass^2(k_\perp)}{2k},\genericmass\right)\gamma^1\widetilde{c}_+\left(\frac{k^2-\widetilde{\genericmass}^2(k_\perp)}{2k},\widetilde{\genericmass}\right),
\end{gather*}
where $\genericmass^2(k_\perp)=\genericmass^2+|k_\perp|^2$ and $r(k,k_\perp,\genericmass)=\frac{k^2+\genericmass^2(k_\perp)}{2k^2}$.
It is then enough to choose $c_+,\widetilde{c}_+$ in such a way that $r(k,k_\perp,\genericmass)
\overline{c_+}\left(\frac{k^2-\beta^2(k_\perp)}{2k},\beta\right)$ is the total derivative in the mass parameter $\beta$ of a smooth function in $\beta$ which, at $\beta =\sup\massinterval$, is equal to a positive fast decreasing function in $k$, while it vanishes at $\beta=\inf\massinterval$.
The integral is thus nonvanishing.

This implies that the SMOP cannot be satisfied and hence it is not possible to construct the Fermionic Projector.
Nonetheless \eqref{mwm} is a sufficient condition for the realization of the mWMOP.
\end{proof}

\textbf{\textit{Acknowledgements.}} 
We would like to thank Claudio Dappiaggi, Felix Finster, Christian G\'erard, Valter Moretti and Nicola Pinamonti for helpful discussions and comments on the manuscript. We are grateful to the referee for useful comments on the manuscript. S.M. is supported within the DFG research training group GRK 1692 ``Curvature, Cycles, and Cohomology'' and he is grateful to the department of mathematics of the University of Genoa for the kind hospitality during part of the realization of this work. N.D. is supported by a Ph.D. grant of the university of Genoa and he is grateful to the department of mathematics of the Universit\'e de Paris-Sud for the kind hospitality during part of the realization of this work.
The work for this paper was partly carried out during the program Modern Theory of Wave Equations at the Erwin Schr\"odinger Institute and we are grateful to the ESI for support and the hospitality during our stay.

\begin{small}

\end{small}


\begin{thebibliography}{99}


\bibitem[Ar71]{araki} 
H.~Araki, \textit{``On quasifree states of CAR and Bogoliubov automorphisms."} Publ.\ Res.\ Inst.\ Math.\ Sci.\ Kyoto {\bf 6}, 385 (1971)

%\bibitem[Ar99]{ArakiB} 
%Araki, H., \textit{``Mathematical theory of quantum fields."} Vol. 101. Oxford University Press on Demand, (1999).

\bibitem[B\"a14]{Bar1} 
C.~B\"ar, \textit{``Green-hyperbolic operators on globally hyperbolic spacetimes.''} Commun.\ Math.\ Phys.\  {\bf 333}, 1585 (2015)

\bibitem[BGP07]{Bar2} 
C.~B\"ar, N. Ginoux and F. Pf\"affle, \textit{``Wave equations on Lorentzian manifolds and quantization.''} European Mathematical Society (2007)

\bibitem[BBSS15]{dual} 
C.~Becker, M., Benini, A. Schenkel and R.J. Szabo, \textit{``Abelian duality on globally hyperbolic spacetimes.''} Comm. Math. Phys. \textbf{349}, 361  (2015) 

\bibitem[BD15]{aqft0}
M.~Benini and C.~Dappiaggi,\textit{``Models of free quantum field theories on curved backgrounds.''} Advances in Algebraic Quantum Field Theory. Springer International Publishing, 75-124 (2015)

\bibitem[BDH13]{aqft1} 
M.~Benini, C.~Dappiaggi and T.~P.~Hack,\textit{``Quantum field theory on curved backgrounds -- a primer."}   Int.\ J.\ Mod.\ Phys.\ A {\bf 28}, 1330023 (2013)

\bibitem[BDM14]{BDM} 
 M.~Benini, C.~Dappiaggi and S.~Murro,
\textit{``Radiative observables for linearized gravity on asymptotically flat spacetimes and their boundary induced states."}   J.\ Math.\ Phys.\  {\bf 55}, 082301 (2014)



\bibitem[BDFY15]{aqft2} 
R.~Brunetti, C. Dappiaggi, K. Fredenhagen and J. Yngvason, \textit{``Advances in algebraic quantum field theory."} Springer (2015)
  
\bibitem[BFV09]{Brunetti e Fredenhagen} 
R.~Brunetti and K.~Fredenhagen, \textit{``Quantum Field Theory on Curved Backgrounds."}
Quantum field theory on curved spacetimes. Springer Berlin Heidelberg, (2009)
  
  \bibitem[BFK96]{BFK} 
R.~Brunetti, K.~Fredenhagen and M.~K\"ohler, \textit{``The microlocal spectrum condition and Wick polynomials of free fields on curved spacetimes.'' } Commun.\ Math.\ Phys.\  {\bf 180}, 633 (1996)
  

\bibitem[DD16]{nico}
  C.~Dappiaggi and N.~Drago,
  \textit{``Constructing Hadamard States via an Extended M\o ller Operator.''}
  Lett.\ Math.\ Phys.\  {\bf 106}, 1587 (2016)

 
 \bibitem[DHP09]{aqft3} 
  C.~Dappiaggi, T.~P.~Hack and N.~Pinamonti,\textit{``The extended algebra of observables for Dirac fields and the trace anomaly of their stress-energy tensor."}   Rev.\ Math.\ Phys.\  {\bf 21}, 1241 (2009)

\bibitem[DHP10]{DHP} 
  C.~Dappiaggi, T.~P.~Hack and N.~Pinamonti, 
  {\it ``Approximate KMS states for scalar and spinor fields in Friedmann-Robertson-Walker spacetimes.''}   Annales Henri Poincare {\bf 12}, 1449 (2011)

\bibitem[DMP06]{DMP} 
C.~Dappiaggi, V.~Moretti and N.~Pinamonti,
{\it ``Rigorous steps towards holography in asymptotically flat spacetimes.''}
Rev. Math. Phys. \textbf{18}, 349 (2006)


\bibitem[DMP09]{DMP2} 
C.~Dappiaggi, V.~Moretti and N.~Pinamonti,
\textit{``Distinguished quantum states in a class of cosmological spacetimes and their Hadamard property.''} J. Math. Phys.  \textbf{50}, 062304 (2009)

\bibitem[DMP11]{DMP3} 
C.~Dappiaggi, V.~Moretti and N.~Pinamonti, \textit{``Rigorous construction and Hadamard property of the Unruh state in Schwarzschild spacetime."} Adv. Theor. Math. Phys. \textbf{15} 355   (2011)
 
  \bibitem[DNP16]{cospi} 
C.~Dappiaggi, G. Nosari and N. Pinamonti, \textit{``The Casimir effect from the point of view of algebraic quantum field theory."}  Math.\ Phys.\ Anal.\ Geom.\ \textbf{19}, 12  (2016)


 \bibitem[DG13]{gerard} 
J.~Derezi\'nski, and C. G\'erard, \textit{``Mathematics of quantization and quantum fields.''} Cambridge University Press, (2013)

     \bibitem[Di80]{Di80}
J.~Dimock,  \textit{``Algebras of local observables on a manifold."} 
 Commun.\ Math.\ Phys.\  \textbf{77}, 219 (1980)
 
  \bibitem[Di82]{dimock}
J.~Dimock, \textit{``Dirac quantum fields on a manifold."}  Trans. Amer. Math. Soc. \textbf{269} 133 (1982)


  \bibitem[DHP17]{nico1}
N.~Drago,  T.-P.~Hack and N.~Pinamonti, \textit{``The generalised principle of perturbative agreement and the thermal mass."}   Annales Henri Poincare {\bf 18}, 807 (2017)
  
  \bibitem[DP14]{nico2}
  N.~Drago and N.~Pinamonti,
  \textit{``Influence of quantum matter fluctuations on geodesic deviation."} 
    J.\ Phys.\ A {\bf 47} 375202 (2014)


\bibitem[DF03]{Dutsch Fredenhagen}
M.~D\"utsch and K. Fredenhagen, \textit{``The Master Ward identity and generalized Schwinger-Dyson equation in classical field theory.''} Comm. Math. Phys. \textbf{243},  275 (2003)

\bibitem[FL15]{ChrisBenni} 
  C.~J.~Fewster and B.~Lang,
 \textit{``Pure quasifree states of the Dirac field from the  projector."} 
  Class.\ Quant.\ Grav.\  {\bf 32}, 095001 (2015)
  
 
   \bibitem[FV02]{FVD} 
  C.~J.~Fewster and R. Verch, \textit{``A Quantum Weak Energy Inequality for Dirac Fields in Curved Spacetime."} Comm. Math. Phys. \textbf{225}, 331 (2002)
 
   \bibitem[FV13]{ChrisVerch} 
  C.~J.~Fewster and R. Verch, \textit{``The necessity of the Hadamard condition."} Class. Quant. Grav. \textbf{30}, 235027 (2013)

 \bibitem[FMR16]{simo} 
F.~Finster, S.~Murro and C.~R\"oken, \textit{``The fermionic projector in a time-dependent external potential: Mass oscillation property and hadamard states."}
 J.\ Math.\ Phys.\ \textbf{57}, 072303 (2016)
 
\bibitem[FiRe15]{Felix1} 
 F.~Finster and M.~Reintjes,
  \textit{``A non-perturbative construction of the fermionic projector on globally hyperbolic manifolds I-Space-times of finite lifetime."}    Adv.\ Theor.\ Math.\ Phys.\  {\bf 19}, 761 (2015)


\bibitem[FiRe16]{Felix2}    
 F.~Finster and M.~Reintjes,
 \textit{``A non-perturbative construction of the fermionic projector on globally hyperbolic manifolds II-space-times of infinite lifetime."}   Adv.\ Theor.\ Math.\ Phys.\  {\bf 20}, 1007 (2016)


\bibitem[FiRe17]{Felix5}    
 F.~Finster and M.~Reintjes,
  \textit{``The Fermionic Signature Operator and Hadamard States in the Presence of a Plane Electromagnetic Wave.''}  Annales Henri Poincare \textbf{18} (2017)
  
  
\bibitem[FrRe16]{FR16}
  K.~Fredenhagen and K.~Rejzner,
 \textit{``Quantum field theory on curved spacetimes: Axiomatic framework and examples.''}    J.\ Math.\ Phys.\  {\bf 57}, 031101 (2016)  



\bibitem[Fu89]{Fulling}
S.~A.~Fulling,
{\it ``Aspects of Quantum Field Theory in Curved Space-time.''}
Cambridge university press, (1989)

\bibitem[FSW78]{Fulling1} 
  S.~A.~Fulling, M.~Sweeny and R.~M.~Wald, \textit{``Singularity structure of the two-point function in quantum field theory in curved spacetime."} Commun.\ Math.\ Phys.\ \textbf{63}, 257 (1978)

  \bibitem[FNW81]{Fulling2} 
  S.~A.~Fulling, F. J. Narcowich and R. M. Wald, \textit{``Singularity structure of the two-point function in quantum field theory in curved spacetime, II."} Annals Phys. \textbf{136}, 243  (1981)  

\bibitem[GW14]{GW1}
C.~G\'erard and M.~Wrochna,
{\it ``Construction of Hadamard states by pseudo-differential calculus.''}
  Commun.\ Math.\ Phys.\  {\bf 325}, 713 (2014)

\bibitem[GW16]{GW2}
C.~G\'erard and M.~Wrochna,
{\it ``Construction of Hadamard states by characteristic Cauchy problem.''}
  Anal. PDE \textbf{9}, 111  (2016) 

 
  \bibitem[GK89]{GK} 
 G.~Gonnella and B.~S.~Kay,\textit{``Can locally Hadamard quantum states have non-local singularities?"}   Class.\ Quant.\ Grav.\  {\bf 6}, 1445 (1989)
 
\bibitem[HS13]{ThomasAlex}
T.-P.~Hack and A.~Schenkel, 
  \textit{``Linear bosonic and fermionic quantum gauge theories on curved spacetimes.''}
Gen. Rel. Grav.   \textbf{45},  877  (2013)
 
 \bibitem[Ha12]{Haag} 
R.~Haag, \textit{ ``Local quantum physics: Fields, particles, algebras."} Springer Science \& Business Media, (2012)
 
 \bibitem[HK64]{HK}
  R.~Haag and D.~Kastler,
\textit{``An algebraic approach to quantum field theory."}   J.\ Math.\ Phys.\  {\bf 5}, 848 (1964) 
 
   \bibitem[HW02]{HW} 
S.~Hollands and R.~M.~Wald,\textit{``Existence of local covariant time ordered products of quantum fields in curved spacetime."}   Commun.\ Math.\ Phys.\  {\bf 231}, 309 (2002)

 
  \bibitem[Is78]{Isham} 
  C.~J.~Isham,
 \textit{``Spinor Fields in Four Dimensional Space-Time."}
  Proc.\ Roy.\ Soc.\ Lond.\ A {\bf 364}, 591 (1978)

\bibitem[KM15]{IgorValter}
I.~Khavkine and V.~Moretti,
  {\it ``Algebraic QFT in Curved Spacetime and quasifree Hadamard states: an introduction.''}  Advances in Algebraic Quantum Field Theory. Springer International Publishing, (2015)
  


\bibitem[LM89]{Lawson} 
H.~B.~Lawson and M.-L.~Michelsohn, \textit{``Spin geometry."}
Vol. 1. Princeton: Princeton university press, (1989)

  \bibitem[La02]{Lax}
 P.~D.~Lax,
  {\it ``Functional Analysis.''}
Wiley-Interscience, (2002)


\bibitem[Mu17]{simoPhD}
S.~Murro,
\textit{``Quantum States on the Algebra of Dirac Fields: A functional analytic approach.''} 
  \url{https://epub.uni-regensburg.de/35661/1/TESI.pdf}

\bibitem[Ni02]{nicolas} 
J.-P.~Nicolas,  \textit{``Dirac fields on asymptotically flat space-times."} Dissertationes Mathematicae \textbf{408} (2002)



\bibitem[Pe52]{Peierls}
  R.~E.~Peierls,
\textit{``The Commutation laws of relativistic field theories.''} 
  Proc.\ Roy.\ Soc.\ Lond.\ A {\bf 214}, 143 (1952)

\bibitem[Ra96a]{Radzikowski} 
M.~J.~Radzikowski, \textit{``Micro-local approach to the Hadamard condition in quantum field theory on curved space-time."}   Commun.\ Math.\ Phys.\  {\bf 179}, 529 (1996)


\bibitem[Ra96b]{Radzikowski2} 
M.~J.~Radzikowski,
{\it ``A Local to global singularity theorem for quantum field theory on curved space-time.''}
  Commun.\ Math.\ Phys.\  {\bf 180}, 1 (1996)

\bibitem[SV01]{SahlmannVerch} 
  H.~Sahlmann and R.~Verch,
\textit{``Microlocal spectrum condition and Hadamard form for vector-valued quantum fields in curved spacetime."} Rev.\ Math.\ Phys.\ \textbf{13}, 1203 (2001)

\bibitem[Sa10]{KOD} 
  K.~Sanders,
\textit{``The locally covariant Dirac field."} Rev.\ Math.\ Phys.\ \textbf{22}, 381 (2010)

\bibitem[Sa13]{KO}
  K.~Sanders,\textit{``A note on spacelike and timelike compactness."} Class. Quant. Grav. \textbf{30}, 115014 (2013).

  \bibitem[SDH14]{SDH12} 
  K.~Sanders, C.~Dappiaggi and T.~P.~Hack,
\textit{  ``Electromagnetism, local covariance, the Aharonov-Bohm effect and Gauss' law.''}
  Commun.\ Math.\ Phys.\  {\bf 328}, 625 (2014)

\bibitem[Ta11]{Taylor}
M.~E.~Taylor, 
\textit{``Partial Differential Equations I.''}
Springer New York (2011)

\bibitem[Wa94]{wald} 
R.~M.~Wald, \textit{``General relativity"}University of Chicago press, (2010)

\bibitem[Wa12]{waldmann} 
W.~Waldmann,  \textit{``Geometric wave equations."} \url{https://arxiv.org/abs/1208.4706} (2012)


\bibitem[Wr12]{Wr} 
M.~Wrochna,  \textit{``Quantum field theory in static external potentials and Hadamard states."} Annales Henri Poincare \textbf{13}, 1841 (2012)

\bibitem[WZ14]{WZ} 
M.~Wrochna and J.~Zahn,\textit{``Classical phase space and Hadamard states in the BRST formalism for gauge field theories on curved spacetime."}  Rev.\ Math.\ Phys.\  {\bf 29}, 1750014 (2017)

\bibitem[Za14]{ZAD} 
J.~Zahn, \textit{``The renormalized locally covariant Dirac field."} Rev.\ Math.\ Phys.\  {\bf 26}, 1330012 (2014)
\end{thebibliography}
\end{document}